\documentclass[journal,twoside,web]{ieeecolor}

\usepackage{generic}
\usepackage{cite}
\usepackage{amsmath,amssymb,amsfonts}
\usepackage{graphicx}
\usepackage{textcomp}

\usepackage{lcsys}

\usepackage{booktabs}
\usepackage{multirow}
\usepackage{subcaption}
\usepackage{mathrsfs}
\usepackage{mathtools}
\usepackage{stmaryrd}
\usepackage{algorithm}
\usepackage[noend]{algpseudocode}
\usepackage{physics}

\algrenewcommand\algorithmicindent{1em}

\algnewcommand\algorithmicparameter{\textbf{Parameter:}}
\algnewcommand\Parameter{\item[\algorithmicparameter]}
\algnewcommand\algorithmiccondition{\textbf{Condition:}}
\algnewcommand\Condition{\item[\algorithmiccondition]}
\makeatletter
\renewcommand{\ALG@name}{Protocol}
\makeatother
\algrenewcommand\algorithmicdo{}

\newcommand{\asmref}[1]{Assumption~\ref{#1}}

\newcommand{\lemref}[1]{Lemma~\ref{#1}}
\newcommand{\thmref}[1]{Theorem~\ref{#1}}

\newcommand{\secref}[1]{Section~\ref{#1}}

\newcommand{\figref}[1]{Fig.~\ref{#1}}
\newcommand{\tabref}[1]{Table~\ref{#1}}

\newcommand{\Z}{\mathbb{Z}}
\newcommand{\Q}{\mathbb{Q}}
\newcommand{\R}{\mathbb{R}}
\newcommand{\ZZ}[1]{\Z_{\langle #1 \rangle}}
\newcommand{\QQ}[2]{\Q_{\langle #1, #2 \rangle}}
\newcommand{\cind}{\overset{c}{\equiv}}
\newcommand{\sind}{\overset{s}{\equiv}}
\newcommand{\pp}{\mathsf{pp}}
\newcommand{\crs}{\mathsf{crs}}
\newcommand{\Share}{\mathsf{Share}}
\newcommand{\Reconst}{\mathsf{Reconst}}
\newcommand{\Mult}{\mathsf{Mult}}
\newcommand{\Setup}{\mathsf{Setup}}
\newcommand{\Offline}{\mathsf{Offline}}
\newcommand{\Online}{\mathsf{Online}}
\newcommand{\View}{\mathsf{View}}
\newcommand{\Sim}{\mathsf{Sim}}

\DeclareMathOperator{\sample}{\gets_{\$}}

\DeclarePairedDelimiter{\floor}{\lfloor}{\rfloor}
\DeclarePairedDelimiter{\share}{\llbracket}{\rrbracket}

\newtheorem{theorem}{Theorem}[section]
\newtheorem{lemma}[theorem]{Lemma}

\newtheorem{definition}[theorem]{Definition}
\newtheorem{remark}[theorem]{Remark}
\newtheorem{assumption}[theorem]{Assumption}

\begin{document}

\def\BibTeX{{\rm B\kern-.05em{\sc i\kern-.025em b}\kern-.08em
    T\kern-.1667em\lower.7ex\hbox{E}\kern-.125emX}}
\markboth{\journalname, VOL. XX, NO. XX, XXXX 2017}
{Teranishi: Secure Two-Party Matrix Multiplication from Lattices and Its Application to Encrypted Control}

\title{Secure Two-Party Matrix Multiplication from Lattices and Its Application to Encrypted Control}

\author{
Kaoru Teranishi
\thanks{Department of Information and Physical Sciences, Graduate School of Information Science and Technology, The University of Osaka, 1-5 Yamadaoka, Suita, Osaka 565-0871, Japan; \texttt{k-teranishi@ist.osaka-u.ac.jp}}
}

\thispagestyle{empty}
\hspace{-4.5mm}
\fbox{
\begin{minipage}{\textwidth-5mm}\scriptsize
© 20XX IEEE.
Personal use of this material is permitted.
Permission from IEEE must be obtained for all other uses, in any current or future media, including reprinting/republishing this material for advertising or promotional purposes, creating new collective works, for resale or redistribution to servers or lists, or reuse of any copyrighted component of this work in other works.
\end{minipage}
}
\newpage
\setcounter{page}{0}

\maketitle
\thispagestyle{empty}

\begin{abstract}
In this study, we propose a two-party computation protocol for approximate matrix multiplication of fixed-point numbers.
The proposed protocol is provably secure under standard lattice-based cryptographic assumptions and enables matrix multiplication at a desired approximation level within a single round of communication.
We demonstrate the feasibility of the protocol by applying it to the secure implementation of a linear control law.
Our evaluation reveals that the client achieves lower online computational complexity compared to the original controller computation, while ensuring the privacy of controller inputs, outputs, and parameters.
Furthermore, a numerical example confirms that the proposed method maintains sufficient precision of control inputs even in the presence of approximation and quantization errors.
\end{abstract}


\section{Introduction}\label{sec:introduction}

\IEEEPARstart{E}{ncrypted} control is a paradigm for ensuring the security and privacy of cyber-physical systems by leveraging cryptographic primitives~\cite{Kogiso2015-go,Darup2021-qq}.
In this framework, control algorithms are executed directly on encrypted data.
This enables a client owning a physical plant to outsource controller computation to untrusted computing resources without disclosing sensitive information, such as sensor measurements, control actions, and controller parameters.

A primary architecture for encrypted control is the single-server setting using homomorphic encryption~\cite{Darup2021-qq}.
Despite significant theoretical advances, in such a setting, the computational burden on the client often exceeds that of the original unencrypted computation, even when implementing simple state-feedback controllers.
This issue stems from the intrinsic computational overhead of the encryption and decryption processes in homomorphic encryption.

To overcome the limitation of the single-server setting, several studies have explored encrypted control in multi-party architectures.
These frameworks usually employ secret sharing to protect sensitive information.
For example, linear control laws were initially implemented in a two-party setting using secret sharing~\cite{Darup2019-ou}.
This approach was subsequently extended to polynomial control laws by integrating homomorphic encryption~\cite{Darup2020-cm}.
The polynomial control was also realized using only secret sharing in settings with more than three parties~\cite{Schlor2021-am}.
Recent advancements have also enabled the implementation of neural-network-based controllers~\cite{Tjell2021-ar} and dynamic controllers~\cite{Teranishi2025-wv} in two-party settings.

A fundamental challenge in encrypted control using secret sharing is the secure evaluation of multiplications between two secret-shared values.
The initial implementation bypassed this issue by assuming that one of the values was public~\cite{Darup2019-ou}.
Furthermore, a naive interactive protocol leveraging homomorphic encryption typically requires multiple rounds of communication between the computing parties~\cite{Darup2020-cm}.
While replicated secret sharing~\cite{Araki2016-pf}, as employed in~\cite{Schlor2021-am}, offers an efficient alternative, it necessitates at least three non-colluding parties, thereby increasing the architectural complexity and deployment costs.
Another common strategy is the use of Beaver triples~\cite{Beaver1991-kr}, as seen in~\cite{Tjell2021-ar,Teranishi2025-wv}.
However, their online generation again requires multi-round communication.
Even if these triples are generated and provided by the client to the computing parties~\cite{Teranishi2025-wv}, this approach inherently increases the computational complexity and communication overhead of the client.
Consequently, simultaneously achieving low client-side complexity in a two-party setting with a single-round communication remains a significant challenge.

To address these challenges, this paper constructs a one-round two-party computation protocol for approximate matrix multiplication. 
Our approach integrates additive secret sharing with lattice-based cryptographic primitives, specifically encryption and commitment schemes. 
By exploiting the linearity of these lattice-based primitives, the proposed protocol enables the computing parties to securely evaluate approximate matrix multiplication within a single round of communication. 
This approach eliminates the need for multiple interactions between the parties. 
Moreover, the framework ensures that the online tasks of the client are restricted to the generation and reconstruction of shares, resulting in an online computational complexity lower than that of the original controller computation.

The main contributions of this study are as follows.
\begin{itemize}
    \item
    We propose a one-round two-party computation protocol for the approximate matrix multiplication of fixed-point numbers.
    The security of this protocol is guaranteed under the Learning With Errors assumption, which is a standard assumption in lattice-based cryptography.

    \item
    We provide a theoretical analysis of the approximation errors inherent in the proposed protocol.
    Sufficient conditions for the precision of fixed-point numbers are derived to guarantee that the approximation error remains below a specified level.

    \item
    By applying the protocol to the implementation of linear control laws, we demonstrate that the proposed framework simultaneously achieves low online computational complexity for the client, one-round communication between the two parties, and the concealment of both signals and feedback gains.
\end{itemize}

\section{Preliminaries}\label{sec:preliminaries}

\subsection{Notations}

Let $\Z$ and $\R$ denote the sets of integers and real numbers, respectively.
For an integer $q \ge 2$, define $\Z_q \coloneqq \Z \cap [-q/2, q/2)$.
For any $x \in \Z$, the reduction of $x$ modulo $q$ is given by $x \bmod q \coloneqq x - \floor{ \frac{x + q / 2}{q} } q$, where $\floor{ \cdot }$ denotes the floor function.
This modular reduction is applied for each element of vectors and matrices.
The maximum and Euclidean norms of a vector or matrix are denoted by $\norm{ \cdot }_{\max}$ and $\norm{ \cdot }$, respectively.
The notation $x \sample X$ indicates that $x$ is sampled either uniformly at random from a finite set $X$ or according to a specified distribution $X$.
Throughout this paper, $\lambda$ denotes the security parameter.
Two distribution ensembles $X$ and $Y$ are computationally indistinguishable, denoted by $X \cind Y$, if no probabilistic polynomial-time algorithm can distinguish them with a non-negligible advantage in $\lambda$.
We write $X \sind Y$ if their statistical distance is negligible in $\lambda$.

\subsection{Secret sharing}

In this study, we employ a two-party secret sharing scheme defined over a matrix message space.

\begin{definition}[Secret sharing]\label{def:ss}
    A secret sharing scheme consists of polynomial-time algorithms $\Share$ and $\Reconst$:
    \begin{itemize}
        \item $\share{X} \gets \Share(X)$:
        The share generation algorithm takes a matrix $X$ over $\Z_q$ as input and outputs shares $\share{X} = (\share{X}_0, \share{X}_1) = (R, X - R \bmod q)$, where $R$ is a random matrix over $\Z_q$ of appropriate dimensions.
        
        \item $X \gets \Reconst(\share{X})$:
        The reconstruction algorithm takes shares $\share{X} = (\share{X}_0, \share{X}_1)$ as input and outputs $X = \share{X}_0 + \share{X}_1 \bmod q$.
    \end{itemize}
\end{definition}

The secret sharing scheme satisfies $\Reconst(\Share(X)) = X$ for any matrix $X$ over $\Z_q$.
Moreover, each share $\share{X}_i$ held by party $P_i$ is uniformly random and thus leaks no information about $X$ to that party.
We respectively define the addition/subtraction and constant multiplication of shares by $\share{X} \pm \share{Y} \coloneqq (\share{X}_0 \pm \share{Y}_0 \bmod q, \share{X}_1 \pm \share{Y}_1 \bmod q)$ and $X \share{Y} \coloneqq (X \share{Y}_0 \bmod q, X \share{Y}_1 \bmod q)$, where $X$ and $Y$ are of compatible dimensions.
It then holds that
\begin{align}
    \share{X} \pm \share{Y} &= \share{ X \pm Y \bmod q }, \label{eq:ss-add} \\
    X \share{Y} &= \share{ X Y \bmod q }. \label{eq:ss-mult}
\end{align}
We also define $\share{X} Y$ as similar to $X \share{Y}$, and the transpose of shares as $\share{X}^\top \coloneqq (\share{X}_0^\top, \share{X}_1^\top)$.
These operations reveal no information about the secrets $X$ and $Y$ since they can be performed locally by each party without any interaction.

\section{Lattice Problems}\label{sec:lattice-problems}

Learning With Errors (LWE)~\cite{Regev2009-ys} and Short Integer Solution (SIS)~\cite{Ajtai1996-hm} are standard cryptographic problems defined over lattices.
These problems serve as the foundation for lattice-based cryptography and are applied to a wide range of applications, including encryption and commitment schemes.
This section provides a brief overview of the definitions and applications of the LWE and SIS problems.

\subsection{Learning with errors}

We begin by defining the LWE problem, which requires distinguishing a structured distribution over a lattice from a uniform one.

\begin{definition}[LWE]
    Let $n$, $m$, and $q$ be positive integers, and let $\chi^d$ be a distribution defined over $\Z^d$.
    The (decision) LWE problem is to distinguish between the two distributions $(A, A^\top s + e \bmod q)$ and $(A, w)$, where $A \sample \Z_q^{n \times m}$, $s \sample \chi^n$, $e \sample \chi^m$, and $w \sample \Z_q^m$.
    The LWE assumption states that these two ensembles are computationally indistinguishable, namely $(A, A^\top s + e \bmod q) \cind (A, w)$.
\end{definition}

In our framework, we sample the secret and error vectors from a discrete Gaussian distribution.
A discrete Gaussian distribution $\chi_\sigma^d$ over $\Z^d$ with a width parameter $\sigma > 0$ is defined by $\Pr[ X = x ] = \rho_\sigma(x) / \sum_{y \in \Z^d} \rho_\sigma(y)$ for each $x \in \Z^d$, where $\rho_\sigma(x) = {\exp}( -\pi \norm{x}^2 / \sigma^2 )$.
In lattice-based cryptography, these distributions are typically treated as bounded, which is formalized as follows.

\begin{assumption}\label{asm:bounded-dist}
    Let $B > 0$.
    For any sample $[x_1 \ \cdots \ x_d]^\top \allowbreak \sample \chi_\sigma^d$, $\abs{x_i} < B$ holds except with negligible probability.
\end{assumption}

For the width $\sigma = 3.2$ commonly adopted in the literature~\cite{Bossuat2024-kz}, the bound $B = 10 \sigma$ is sufficient to satisfy the negligible probability requirement in practice.
Given $\chi^n = \chi_\sigma^n$ and $\chi^m = \chi_\sigma^m$, the hardness of the LWE problem primarily depends on the dimension $n$, the modulus $q$, and the width parameter $\sigma$.
In general, the problem becomes more difficult as $n$ or the ratio $\sigma / q$ increases.
Because the number of samples $m$ provided to a challenger does not affect the underlying hardness, we refer to $(n, q, \sigma)$ as the LWE parameters.

The LWE problem is of significant interest because it provides a template for encryption schemes~\cite{Regev2009-ys}.
For instance, a message $\mu \in \Z_q^m$ can be encrypted using a secret key $s$ as
\begin{equation}
    c = A^\top s + \mu + e \bmod q.
    \label{eq:lwe-enc}
\end{equation}
Under the LWE assumption, a polynomial-time adversary learns nothing about $\mu$ from $c$.
Furthermore, decryption is performed via $[A^\top \ c] [-s^\top \ 1]^\top = \mu + e \bmod q \approx \mu$, which allows for the approximate recovery of the original message.

\subsection{Short integer solution}

We next introduce the SIS problem, which requires finding a short nonzero vector satisfying a linear equation over a lattice, given a random coefficient matrix and a random bias.

\begin{definition}[SIS]
    Let $n$, $t$, and $q$ be positive integers, and let $\beta > 0$ be a real number.
    The (inhomogeneous) SIS problem is to find a nonzero vector $r \in \Z^t$ such that $h = A r \bmod q$ with $\norm{r} \le \beta$ for some norm $\norm{\cdot}$, given $A \sample \Z_q^{n \times t}$ and $h \sample \Z_q^n$.
    The SIS assumption states that no polynomial-time algorithm can find such a short vector except with negligible probability.
\end{definition}

The hardness of the SIS problem depends on the dimensions $n$ and $t$, the modulus $q$, and the norm bound $\beta$.
For our purposes, $n$ and $q$ are chosen to be consistent with the LWE parameters, while $\beta$ is determined by protocol construction.
Consequently, we refer to $t$ as the SIS parameter.

The SIS problem can be employed to construct commitment schemes satisfying the hiding and binding properties.
The former ensures that the commitment reveals no information about the underlying message, while the latter prevents the sender from later opening the commitment to a different message.
Specifically, a commitment to a message $\mu \in \Z_q^m$ is computed as
\begin{equation}
    h = A \mu + B r \bmod q,
    \label{eq:sis-cmt}
\end{equation}
where $A \sample \Z_q^{n \times m}$, $B \sample \Z_q^{n \times t}$, and $r \sample \Z_3^t$~\cite{Ajtai1996-hm}.
By the Leftover Hash Lemma~\cite{Impagliazzo1989-fn}, $(B, B r \bmod q) \sind (B, w)$ holds for $w \sample \Z_q^n$ if $t \log_2 3 \ge n \log_2 q + 2 \lambda$, which guarantees the hiding property.
Moreover, the SIS assumption implies that no polynomial-time adversary can find a distinct pair $(\mu', r')$ such that $h = A \mu' + B r' \bmod q = [A \ B] [(\mu')^\top \ (r')^\top]^\top \bmod q$ with $\mu' \ne \mu$, thereby satisfying the binding property.

\section{Two-Party Matrix Multiplication}\label{sec:protocol}

This section proposes a two-party matrix multiplication protocol.
The protocol is built based on an approximate inner-product computation from the LWE and SIS problems.
We first revisit this underlying computation and then present the detailed construction of the proposed protocol.
Furthermore, we derive the conditions for the bit lengths of input data to guarantee the desired precision of the results.
Finally, we provide a security analysis of the protocol.

\subsection{Approximate inner product from LWE and SIS}

The approximate inner-product computation using LWE ciphertexts and SIS commitments is introduced in~\cite{Abram2024-me,Boyle2025-hx}.
Let $(n ,q, \sigma)$ and $t$ be the LWE and SIS parameters, respectively.
Suppose we have LWE ciphertexts \eqref{eq:lwe-enc} of a vector $x \in \Z_q^m$ and the zero vector, and an SIS commitment \eqref{eq:sis-cmt} to a vector $y \in \Z_q^m$, that is,
\begin{align*}
    c &= A^\top s + x + e \bmod q, \quad c' = B^\top s + e' \bmod q, \\
    h &= A y + B r \bmod q,
\end{align*}
where $A \sample \Z_q^{n \times m}$, $B \sample \Z_q^{n \times t}$, $s \sample \chi_\sigma^n$, $e \sample \chi_\sigma^m$, $e' \sample \allowbreak \chi_\sigma^t$, and $r \sample \Z_3^t$.
We now consider the inner products $c^\top y = x^\top y + s^\top A y + e^\top y$, $(c')^\top r = s^\top B r + (e')^\top r$, and $s^\top h = s^\top A y + s^\top B r$.
By subtracting the third from the sum of first two, we can observe that the terms $s^\top A y$ and $s^\top B r$ are eliminated.
This yields an approximate inner product
\begin{align}
    c^\top y + (c')^\top r - s^\top h \bmod q
    &= x^\top y + (\text{noise}) \bmod q, \nonumber \\
    &\approx x^\top y \bmod q, \label{eq:approx-inner}
\end{align}
where the term $e^\top y + (e')^\top r$ constitutes the noise.

The previous studies~\cite{Abram2024-me,Boyle2025-hx} have utilized the property \eqref{eq:approx-inner} to compute shares of an inner product in the scenarios where the two parties respectively hold $x$ and $y$.
By contrast, our focus is on computing the inner product while concealing $x$ and $y$ from the parties, ensuring that only the client knows these values.
The core observation to achieve this is that the left-hand side of \eqref{eq:approx-inner} is linear in $y$, $r$, and $s$, which are the secret components of the LWE ciphertext and the SIS commitment.
Therefore, by leveraging the properties \eqref{eq:ss-add} and \eqref{eq:ss-mult} of the secret sharing scheme, the approximate inner product can be computed over the shares of the secrets.
This approach suggests a method for outsourcing the inner-product computation to two untrusted parties.
Building on this idea, the next section constructs a multiplication protocol tailored for matrix messages of fixed-point numbers.

\subsection{Proposed protocol}

The approximate inner product \eqref{eq:approx-inner} can be extended to the multiplication of $d_1 \times d_2$ and $d_2 \times d_3$ matrices.
Let $x_i, y_j \in \Z_q^{d_2}$ for $i = 1, \dots, d_1$ and $j = 1, \dots, d_3$.
For each vector, we compute the LWE ciphertexts \eqref{eq:lwe-enc} and the SIS commitments~\eqref{eq:sis-cmt} as $c_i = A^\top s_i + x_i + e_i \bmod q$, $c'_i = B^\top s_i + e'_i \bmod q$, and $h_j = A y_j + B r_j \bmod q$, where $A \sample \Z_q^{n \times d_2}$, $B \sample \Z_q^{n \times t}$, $s_i \sample \chi_\sigma^n$, $e_i \sample \chi_\sigma^{d_2}$, $e'_i \sample \chi_\sigma^t$, and $r_j \sample \Z_3^t$.
By arranging these vectors as columns, we obtain the following matrix representations,
\begin{align*}
    C &= A^\top S + X^\top + E \bmod q, \quad C' = B^\top S + E' \bmod q, \\
    H &= A Y + B R \bmod q,
\end{align*}
where the uppercase letters denote the horizontal concatenation of their lowercase counterparts, e.g., $S = [s_1 \ \cdots \ s_{d_1}]$.
Note that $X^\top = [x_1 \ \cdots \ x_{d_1}] \in \Z_q^{d_2 \times d_1}$ appears in a transposed form within $C$ to ensure dimensional compatibility.
As a result, these matrix forms satisfy the approximate relation
\begin{equation}
    C^\top Y + (C')^\top R - S^\top H \bmod q \approx X Y \bmod q
    \label{eq:approx-matmult}
\end{equation}
analogous to the vector case in \eqref{eq:approx-inner}.

Although the approximate matrix multiplication \eqref{eq:approx-matmult} is defined over $\Z_q$, practical applications, including control systems, require computations to be performed over real-valued matrices.
To bridge this gap, we introduce a fixed-point number representation.
Let $\QQ{k}{\ell} \coloneqq \{ 2^{-\ell} z \mid z \in \ZZ{k} \}$ denote the set of $k$-bit fixed-point numbers with an $\ell$-bit fractional part, where $\ZZ{k} \coloneqq \{ -2^{k - 1}, \dots, 2^{k - 1} - 1 \}$ is the set of $k$-bit integers.
A matrix $X \in \QQ{k}{\ell}^{d_1 \times d_2}$ is encoded into an integer matrix $\bar{X} \in \ZZ{k}^{d_1 \times d_2}$ by  scaling with $2^\ell$, i.e., $\bar{X} = 2^\ell X$.
By replacing $X$ and $Y$ in \eqref{eq:approx-matmult} with their encoded values $\bar{X} = 2^\ell X$ and $\bar{Y} = 2^\ell Y$, we obtain an approximate matrix multiplication framework over fixed-point numbers.

Consequently, we propose a two-party matrix multiplication protocol in \figref{fig:protocol}.
The proposed $\Mult$ protocol consists of three subroutines: $\Setup$, $\Offline$, and $\Online$.
The $\Setup$ and $\Offline$ procedures are performed once during an offline phase, whereas the $\mathsf{Online}$ procedure can be executed repeatedly for various inputs during an online phase.

In the offline phase, a client invokes $\Mult.\Setup$ with public parameters $\pp$ to generate random matrices $A \sample \Z_q^{n \times d_2}$ and $B \sample \Z_q^{n \times t}$, and publishes a common reference string $\crs = (\pp, A, B)$ to the parties.
The client then invokes $\Mult.\Offline$ with $\crs$ and the input matrix $X \in \QQ{k}{\ell}^{d_1 \times d_2}$.
During this procedure, the client computes the LWE ciphertexts $C = A^\top S + \bar{X}^\top + E \bmod q$ and $C' = B^\top S + E' \bmod q$, where $\bar{X} = 2^\ell X \in \ZZ{k}^{d_1 \times d_2}$, $S \sample \chi_\sigma^{n \times d_1}$, $E \sample \chi_\sigma^{d_2 \times d_1}$, and $E' \sample \chi_\sigma^{t \times d_1}$.
Finally, it also generates the shares $\share{S}$ of the secret matrix $S$ and distributes $(C, C', \share{S}_i)$ to the party $P_i$.

In the online phase, the client and the two parties invoke $\Mult.\Online$.
The client first encodes the input matrix $Y \in \QQ{k}{\ell}^{d_2 \times d_3}$ into $\bar{Y} = 2^\ell Y \in \ZZ{k}^{d_2 \times d_3}$.
It then generates the shares $\share{\bar{Y}}$ and sends $\share{\bar{Y}}_i$ to the party $P_i$.
Each party $P_i$ locally samples $\share{R}_i \sample \Z_3^{t \times d_3}$ to constitute the random shares $\share{R} = (\share{R}_0, \share{R}_1)$.
The parties then locally compute $\share{H} = A \share{\bar{Y}} + B \share{R}$ using the properties \eqref{eq:ss-add} and \eqref{eq:ss-mult}.
Subsequently, they open the shares to reconstruct $H = \Reconst(\share{H})$.
Using the LWE ciphertexts $C, C'$ and the SIS commitment $H$, the parties compute $\share{\bar{Z}} = C^\top \share{\bar{Y}} + (C')^\top \share{R} - \share{S}^\top H$ locally and returns the resulting shares to the client.
Upon receiving these shares, the client computes $Z = 2^{-2\ell} \Reconst(\share{\bar{Z}})$.

\begin{figure}[t]
    \centering
    
    \setlength{\fboxsep}{1eM}
    \setlength{\fboxrule}{1pt}
    
    \fbox{\begin{minipage}{\linewidth-4eM}
    \leftline{Public parameters: $\pp = (k, \ell, d_1, d_2, d_3, n, q, \sigma, t)$}

    \medskip
    
    \leftline{$\Mult.\Setup(\pp)$:}
    \begin{algorithmic}[1]
        \State $A \sample \Z_q^{n \times d_2}$, $B \sample \Z_q^{n \times t}$
        \State Output $\crs \gets (\pp, A, B)$
    \end{algorithmic}

    \medskip
    
    \leftline{$\Mult.\Offline(\crs, X)$:}
    \begin{algorithmic}[1]
        \State $\bar{X} \gets 2^\ell X$
        \State $S \sample \chi_\sigma^{n \times d_1}$, $E \sample \chi_\sigma^{d_2 \times d_1}$, $E' \sample \chi_\sigma^{t \times d_1}$
        \State $C \gets A^\top S + \bar{X}^\top + E \bmod q$
        \State $C' \gets B^\top S + E' \bmod q$
        \State $\share{S} \gets \Share(S)$
        \State Output $(C, C', \share{S})$
    \end{algorithmic}

    \medskip

    \leftline{$\Mult.\Online(\crs, C, C', \share{S}, Y)$:}
    \begin{algorithmic}[1]
        \State $\bar{Y} \gets 2^\ell Y$
        \State $\share{\bar{Y}} \gets \Share(\bar{Y})$
        \State $\share{R}_0 \sample \Z_3^{t \times d_3}$, $\share{R}_1 \sample \Z_3^{t \times d_3}$
        \State $\share{R} \gets (\share{R}_0, \share{R}_1)$
        \State $\share{H} \gets A \share{\bar{Y}} + B \share{R}$
        \State $H \gets \Reconst(\share{H})$
        \State $\share{\bar{Z}} \gets C^\top \share{\bar{Y}} + (C')^\top \share{R} - \share{S}^\top H$
        \State Output $Z \gets 2^{-2 \ell} \Reconst(\share{\bar{Z}})$
    \end{algorithmic}
    \end{minipage}}

    \caption{Two-party matrix multiplication protocol $\Mult$ consisting of three subroutines: $\Setup$, $\Offline$, and $\Online$.}
    \label{fig:protocol}
    \vspace{-1.5\baselineskip}
\end{figure}

We now establish the approximate correctness of the proposed protocol.
The correctness relies on the noise term inherent in the computation of \eqref{eq:approx-matmult}.
To ensure that the protocol output $Z$ accurately approximates the desired matrix product $XY$, we need to show that the reconstructed value $\bar{Z}$ does not suffer from modular wrap-around (i.e., overflow) and that the magnitude of the noise term remains sufficiently small.
The following lemma guarantees that such overflow is avoided when the bit length of the inputs is bounded according to the cryptographic parameters and matrix dimensions.

\begin{lemma}\label{lem:k-bound}
    Consider the $\Mult$ protocol in \figref{fig:protocol}.
    Under \asmref{asm:bounded-dist} with $\sigma=3.2$ and $B = 10 \sigma$, if
    \begin{equation}
        6 \le k < \frac{1}{2} \log_2 \frac{ q - 128 t }{ d_2 },
        \label{eq:k-bound}
    \end{equation}
    $\bar{Z} = \bar{X} \bar{Y} + E \bar{Y} + E' R$ holds for all $X \in \QQ{k}{\ell}^{d_1 \times d_2}$ and $Y \in \QQ{k}{\ell}^{d_2 \times d_3}$, where $\bar{Z} = \Reconst(\share{\bar{Z}})$ and $R = \Reconst(\share{R})$.
\end{lemma}

\begin{proof}
    The properties \eqref{eq:ss-add} and \eqref{eq:ss-mult} of the secret sharing scheme imply that $\share{\bar{Z}} = \share{ C^\top \bar{Y} + (C')^\top R - S^\top H \bmod q } = \share{ \bar{X} \bar{Y} + E^\top \bar{Y} + (E')^\top R \bmod q }$.
    Let $\hat{Z} = \bar{X} \bar{Y} + E^\top \bar{Y} + (E')^\top R$.
    If $\hat{Z}$ belongs to $\Z_q^{d_1 \times d_3}$, its reduction modulo $q$ satisfies $\hat{Z} \bmod q = \hat{Z}$.
    In this case, the reconstruction implies $\bar{Z} = \Reconst(\share{\bar{Z}}) = \Reconst(\share{\hat{Z}}) = \hat{Z}$.
    
    The maximum norm of $\hat{Z}$ is bounded by $\norm\big{\hat{Z}}_{\max} \le \norm{ \bar{X} \bar{Y} }_{\max} + \norm{ E^\top \bar{Y} }_{\max} + \norm{ (E')^\top R }_{\max} \le 2^{ 2 (k - 1) } d_2 + 2^{k - 1} B d_2 + 2 B t$.
    For $k \ge 6$, we have $2^{ 2 (k - 1) } d_2 + 2^{k - 1} B d_2 + 2 B t \le 2^{2k - 1} d_2 + 2^6 t$, where $2^{k - 1} \ge B = 2^5$.
    Hence, the condition $\hat{Z} \in \Z_q^{d_1 \times d_3}$ holds if $2^{2k - 1} d_2 + 2^6 t < q / 2$.
    Multiplying by two and taking the logarithm, we obtain $2k + \log_2 d_2 < \log_2 (q - 128 t)$, which implies \eqref{eq:k-bound}.
\end{proof}

While \lemref{lem:k-bound} establishes an upper bound on the data size to prevent modular wrap-around, the following theorem provides a lower bound on the fractional bit length required to achieve a desired level of precision.

\begin{theorem}\label{thm:l-bound}
    Consider the $\Mult$ protocol in \figref{fig:protocol}.
    Suppose the condition \eqref{eq:k-bound} is satisfied.
    Under \asmref{asm:bounded-dist} with $\sigma=3.2$ and $B = 10 \sigma$, for every $\epsilon > 0$, if
    \begin{equation}
        \ell > \frac{1}{2} \qty[ k + 4 + \log_2 \frac{ d_2 + t }{ \epsilon } ],
        \label{eq:l-bound}
    \end{equation}
    $\norm{ XY - Z }_{\max} < \epsilon$ holds for all $X \in \QQ{k}{\ell}^{d_1 \times d_2}$ and $Y \in \QQ{k}{\ell}^{d_2 \times d_3}$.
\end{theorem}

\begin{proof}
    Since the condition \eqref{eq:k-bound} is satisfied, \lemref{lem:k-bound} guarantees that $Z = XY + 2^{-2 \ell} \qty( E^\top \bar{Y} + (E')^\top R )$.
    Thus, the maximum norm is bounded by $\norm{ XY - Z }_{\max} \le 2^{-2 \ell} \norm{ E^\top \bar{Y} + (E')^\top R }_{\max} \le 2^{-2 \ell} \qty( 2^{k - 1} B d_2 + 2 B t  ) \le 2^{k - 2 \ell + 4} ( d_2 + t )$, where $2^{k - 1} \ge B$ and $B = 2^5$ are applied.
    This implies that $\norm{ XY - Z }_{\max} < \epsilon$ holds if $2^{k - 2 \ell + 4} ( d_2 + t ) < \epsilon$.
    Taking the logarithm of both sides, we obtain \eqref{eq:l-bound}.
    This completes the proof.
\end{proof}

From \eqref{eq:l-bound}, achieving a smaller $\epsilon$ requires a larger $\ell$, inherently increasing $k$.
Avoiding overflow via \eqref{eq:k-bound} necessitates a larger $q$.
To maintain the security level, $n$ must increase, requiring $t$ to also increase to satisfy \eqref{eq:t-bound} in \thmref{thm:security} of \secref{sec:security}.
As these enlarged parameters increase computational and communication complexities (see \tabref{tab:comparison}), there is a trade-off between precision and implementation costs.

\subsection{Security Analysis}\label{sec:security}

We next analyze the security of the proposed protocol via the simulation paradigm~\cite{Goldreich2009-ct}.
For this analysis, we assume the following adversary model.

\begin{assumption}\label{asm:adversary}
    The computing parties $P_0$ and $P_1$ are semi-honest and non-colluding.
\end{assumption}

Under this model, the computing parties follow the protocol but may individually attempt to infer private information.
This setting is practical when $P_0$ and $P_1$ are distinct and competing cloud providers contracted by the client.
However, this assumption excludes malicious adversaries who actively deviate from the protocol.
Addressing these threats remains an important future direction.

In cryptography, the information available to a party is referred to as its view, consisting of its inputs, internal randomness, and all received messages.
Within the simulation paradigm, security is established by constructing a simulator that generates a view indistinguishable from the real one.

\begin{theorem}\label{thm:security}
    Let $\View_i(\pp, X, Y)$ be the view of $P_i$ during the execution of the $\Mult$ protocol in \figref{fig:protocol}.
    Under the LWE assumption and \asmref{asm:adversary}, if
    \begin{equation}
        t \log_2 3 \ge n \log_2 q + 2 (\lambda + \log_2 d_3),
        \label{eq:t-bound}
    \end{equation}
    then for every $i \in \qty{0, 1}$, there exists a probabilistic polynomial-time simulator $\Sim_i$ such that $\View_i(\pp, X, Y) \allowbreak \cind \Sim_i(\pp)$ for all $X \in \QQ{k}{\ell}^{d_1 \times d_2}$ and $Y \in \QQ{k}{\ell}^{d_2 \times d_3}$.
\end{theorem}

\begin{proof}
    The view of $P_i$ is given by $\View_i(\pp, X, Y) = (\crs, C, C', \share{S}_i, \share{\bar{Y}}_i, \share{R}_i, \share{H}_{1-i})$, where $\crs = (\pp, A, \allowbreak B)$.
    We construct a simulator $\Sim_i$ as follows.
    It takes $\pp$ as input, samples $\mathcal{A} \sample \Z_q^{n \times d_2}$, $\mathcal{B} \sample \Z_q^{n \times t}$, $\mathcal{C} \sample \Z_q^{d_2 \times d_1}$, $\mathcal{C}' \sample \Z_q^{t \times d_1}$, $\share{\mathcal{S}}_i \sample \Z_q^{n \times d_1}$, $\share{\bar{\mathcal{Y}}}_i \sample \Z_q^{d_2 \times d_3}$, $\share{\mathcal{R}}_i \sample \allowbreak \Z_3^{t \times d_3}$, and $\share{\mathcal{H}}_{1-i} \allowbreak \sample \allowbreak \Z_q^{n \times d_3}$, and outputs $((\pp, \mathcal{A}, \mathcal{B}), \mathcal{C}, \mathcal{C}', \allowbreak \share{\mathcal{S}}_i, \share{\bar{\mathcal{Y}}}_i, \share{\mathcal{R}}_i, \share{\mathcal{H}}_{1-i})$.
    We show that the simulator output is computationally indistinguishable from the real view.
    
    In $\Mult.\Setup$, $A \in \Z_q^{n \times d_2}$ and $B \in \Z_q^{n \times t}$ are sampled uniformly at random.
    Hence, the distributions of $(\pp, A, B)$ and $(\pp, \mathcal{A}, \mathcal{B})$ are identical.
    
    Next, consider $C = A^\top S + \bar{X}^\top + E \bmod q$ and $C' = B^\top S + E' \bmod q$.
    Let $\check{A} = [A \ B]$, $\check{E} = [E^\top \ (E')^\top]^\top$, and $\check{\mathcal{C}} = [\mathcal{C}^\top \ (\mathcal{C}')^\top]^\top$.
    For each $j = 1,\dots,d_1$, let $s_j$, $\check{e}_j$, and $\check{c}_j$ be the $j$th columns of $S$, $\check{E}$, and $\check{\mathcal{C}}$, respectively. 
    Each column of $\check{A}^\top S + \check{E} \bmod q$ is an LWE sample, and the LWE assumption implies $(\check{A}, \check{A}^\top s_j + \check{e}_j \bmod q) \cind (\check{A}, \check{c}_j)$.
    Adding the fixed message term to the LWE sample preserves its computational indistinguishability.
    Hence, applying this column by column, we obtain $(A, B, C, C') \cind (A, B, \mathcal{C}, \mathcal{C}')$.
    
    The shares $\share S_i$ and $\share{\bar{Y}}_i$ are uniformly random and independent of $S$ and $\bar{Y}$, respectively.
    Moreover, $\share{R}_i$ is sampled from $\Z_3^{t \times d_3}$ uniformly at random in $\Mult.\Online$.
    Therefore, the joint distributions of $(\share{S}_i, \share{\bar{Y}}_i, \share{R}_i)$ and $(\share{\mathcal{S}}_i, \share{\bar{\mathcal{Y}}}_i, \share{\mathcal{R}}_i)$ are identical.
    
    Finally, consider $\share{H}_{1-i} = A \share{\bar{Y}}_{1-i} + B \share{R}_{1-i} \allowbreak \bmod q$.
    For each $j = 1, \dots, d_3$, let $r_j$ and $h_j$ be the $j$th columns of $\share{R}_{1-i}$ and $\share{\mathcal{H}}_{1-i}$, respectively.
    By the Leftover Hash Lemma~\cite{Impagliazzo1989-fn}, $(B, B r_j \bmod q) \sind (B, h_j)$ holds, provided that $t \log_2 3 \ge n \log_2 q + 2(\lambda + \log_2 d_3)$.
    Applying this to all $d_3$ columns and using a union bound yields $(B, B \share{R}_{1-i} \bmod q) \sind (B, \share{\mathcal{H}}_{1-i})$.
    Adding the fixed term $A \share{\bar{Y}}_{1-i}$ preserves the statistical distance modulo $q$, and therefore $(A, B, \share{H}_{1-i}) \sind (A, B, \share{\mathcal{H}}_{1-i})$.
    
    Combining the above replacements, we obtain $\View_i(\pp, \allowbreak X, \allowbreak Y) \allowbreak \cind \allowbreak \Sim_i(\pp)$.
    This completes the proof.
\end{proof}

\thmref{thm:security} implies that the view of each computing party can be simulated from the public parameters alone.
Therefore, a semi-honest and non-colluding computing party learns no information about $X$, $Y$, or $Z$ beyond what is revealed by the public parameters.
Note that this security proof does not rely on the SIS assumption because it only requires the hiding property of the SIS commitment.

\section{Application to Encrypted Control}\label{sec:control}

In this section, we present an application of the proposed protocol to encrypted control.
To this end, consider a linear control law of the form
\begin{equation}
    u(\tau) = K ( x(\tau) - v(\tau) ),
    \label{eq:control}
\end{equation}
where $\tau = 0, 1, 2, \dots$ is the time index, $u(\tau) \in \R^{n_u}$ is the control input, $x(\tau) \in \R^{n_x}$ is the system state, $v(\tau) \in \R^{n_u}$ is the reference input, and $K \in \R^{n_u \times n_x}$ is the feedback gain.
We assume a scenario where a client wishes to outsource the computation of \eqref{eq:control} to two untrusted parties while concealing $u(\tau)$, $x(\tau)$, $v(\tau)$, and $K$ from them.
In this setting, the feedback gain $K$ is designed by the client, and the reference input $v(\tau)$ is provided by an external operator, while the state $x(\tau)$ originates from a system owned by the client.

Let
\begin{equation}
    \tilde{u}(\tau) \coloneqq \tilde{K} ( \tilde{x}(\tau) - \tilde{v}(\tau) )
    \label{eq:encoded-control}
\end{equation}
be a quantized control input, where $\tilde{K} \in \QQ{k}{\ell}^{n_u \times n_x}$, $\tilde{x}(\tau) \in \QQ{k}{\ell}^{n_x}$, and $\tilde{v}(\tau) \in \QQ{k}{\ell}^{n_u}$ are the fixed-point number representations of $K$, $x(\tau)$, and $v(\tau)$, respectively.
These representations can be obtained using a typical mid-tread quantizer with a step size of $2^{-\ell}$ and a saturation threshold of $2^{k - 1}$.

To apply the $\Mult$ protocol to the controller computation \eqref{eq:encoded-control}, we assign $\tilde{K}$ as $X$ in $\Mult.\Offline$ and slightly modify $\Mult.\Online$.
The modified $\Mult.\Online$ takes $(\crs, C, C', \share{S}, \share{\bar{x}(\tau)}, \share{\bar{v}(\tau)})$ as input and computes $\share{\bar{Y}} = \share{\bar{x}(\tau)} - \share{\bar{v}(\tau)}$, where $\bar{x}(\tau) = 2^\ell \tilde{x}(\tau)$, $\bar{v}(\tau) = 2^\ell \tilde{v}(\tau)$, and $C$ is the LWE ciphertext of $\bar{K} = 2^\ell \tilde{K}$.
At each time step $\tau$, the client and the operator simultaneously send the shares $\share{\bar{x}(\tau)}_i$ and $\share{\bar{v}(\tau)}_i$ to each party $P_i$, respectively.
The parties compute the shares $\share{\bar{Z}}$ according to the modified protocol and return them to the client.
The client then recovers an approximate control input
\begin{equation}
    \hat{u}(\tau) \coloneqq 2^{-2 \ell} \Reconst(\share{\bar{Z}})
    \label{eq:approximate-control}
\end{equation}
from the shares.

\begin{table*}
    \centering
    \caption{Comparison of computational complexity, communication complexity, and required storage size.}
    \label{tab:comparison}
    \begin{tabular}{ccccccc}
        \toprule
        \multirow{2}{*}{Scheme} & \multicolumn{2}{c}{Computational complexity} & \multicolumn{3}{c}{Communication complexity} & Storage size \\
        & Client & Parties & Client $\to$ Parties & Parties $\to$ Client & Party $\leftrightarrow$ Party (1 round) & Parties \\
        \midrule
        Proposed & $O(n_u + n_x)$ & $O(n^2 \log_2 q)$ & $O(n_x \log_2 q)$ & $O(n_u \log_2 q)$ & $O(n \log_2 q)$ & $O((n \log_2 q)^2)$ \\
        \cite{Teranishi2025-wv} & $O(n_u n_x)$ & $O(n_u n_x)$ & $O(n_u n_x \log_2 q)$ & $O(n_u \log_2 q)$ & $O(n_u n_x \log_2 q)$ & $O(n_u n_x \log_2 q)$ \\
        \cite{Kim2023-pk} & $O((n_u + n_x) n)$ & $O(n_u n_x \zeta n^2)$ & $O(n_x n \log_2 q)$ & $O(n_u n \log_2 q)$ & -- & $O(n_u n_x \zeta n^2 \log_2 q)$ \\
        \bottomrule
    \end{tabular}
    \vspace{-1.5\baselineskip}
\end{table*}

\begin{remark}
    \tabref{tab:comparison} compares the online complexities and the required storage size of the proposed scheme with those of the client-aided two-party computation-based scheme~\cite{Teranishi2025-wv} and the Gentry-Sahai-Waters homomorphic encryption-based scheme~\cite{Kim2023-pk}, where $n_u, n_x \ll n$, $t = 2n \log_2 q$, and $\zeta$ is a positive integer.
    The proposed scheme limits the client's computational complexity to $O(n_u + n_x)$, which is lower than the $O(n_u n_x)$ required for the direct evaluation of \eqref{eq:control}, and reduces client-to-party communication.
    These advantages come at the cost of increased complexities and storage size for the computing parties.
    However, to the best of our knowledge, this is the first realization of encrypted control that simultaneously achieves reduced online complexity for the client and confidentiality of both signals and gains in one-round two-party computation.
\end{remark}

To provide a proof of concept, we consider the centralized formation control of $50$ agents.
Each agent is modeled as a discrete-time linear time-invariant system $x_i(\tau + 1) = x_i(\tau) + 0.01 u_i(\tau) + \xi_i(\tau)$, where $x_i(\tau), u_i(\tau), \xi_i(\tau) \in \R^2$ denote the position, input, and noise of the $i$th agent, respectively.
The noise samples are drawn i.i.d. from a zero-mean Gaussian distribution with variance $0.01 I$.
Here, $u(\tau)$, $x(\tau)$, and $v(\tau)$ in \eqref{eq:control} are given as the stacked vectors of $u_i(\tau)$, $x_i(\tau)$, and $v_i(\tau)$, respectively, where $v_i(\tau) = x_i^\mathrm{des} \in \R^2$ denotes the position of the $i$th agent in the target formation.

We design the feedback gain $K \in \R^{100 \times 100}$ in \eqref{eq:control} by solving a linear quadratic regulator problem to minimize the cost function $\lim_{N \to \infty} (1 / N) \sum_{\tau=0}^{N-1} \mathbb{E} [ \sum_{i=1}^{50} \norm{ \eta_i(\tau) }^2 + \sum_{i=1}^{50} \norm{ u_i(\tau) }^2 + \sum_{i < j} \norm{ \eta_i(\tau) - \eta_j(\tau) }^2 ]$, where $\eta_i(\tau) = x_i(\tau) - x_i^\mathrm{des}$.
To achieve a security level of $\lambda = 128$ bits, we choose the LWE parameters $(n, q, \sigma) = (2^{12}, 2^{108}, 3.2)$ and the standard SIS parameter $t = 2n \log_2 q$~\cite{Ajtai1996-hm}, which satisfy \eqref{eq:t-bound}.
For these parameters, \lemref{lem:k-bound} and \thmref{thm:l-bound} guarantee that if the bit lengths $k$ and $\ell$ are chosen such that $42.3 < \ell < k < 50.6$, the approximation error $\norm*{\tilde{u}(\tau) - \hat{u}(\tau)}_{\max}$ remains below $\epsilon = 2^{-10}$.
Accordingly, we set $k = 50$ and $\ell = 43$.

\figref{fig:example}(\subref{fig:formation}) shows the agent trajectories under the control \eqref{eq:control} for $\tau = 0, \dots, 100$ from a grid to a circular formation.
\figref{fig:example}(\subref{fig:error}) then shows the control input errors.
The blue solid line depicts $\| \tilde{u}(\tau) - \hat{u}(\tau) \|_{\max}$, while the red dashed line depicts $\| u(\tau) - \hat{u}(\tau) \|_{\max}$, where $u(\tau)$, $\tilde{u}(\tau)$, and $\hat{u}(\tau)$ are computed using \eqref{eq:control}, \eqref{eq:encoded-control}, and \eqref{eq:approximate-control}, respectively.
The behavior of the blue solid line supports \thmref{thm:l-bound}, as the error remains below the theoretical bound $\epsilon$.
Furthermore, the red dashed line indicates that the encrypted control preserves the original control performance, as the error remains sufficiently small despite the presence of quantization errors.

\begin{figure}
    \vspace{-5mm}
    \centering
    \begin{subfigure}{0.49\columnwidth}
        \centering
        \includegraphics[scale=0.62]{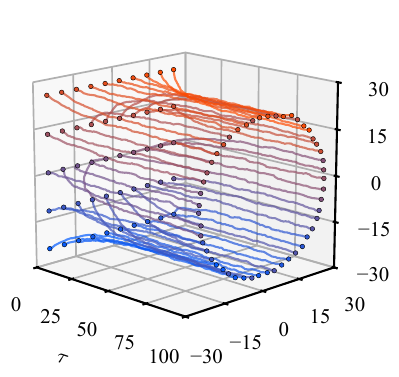}
        \caption{}
        \label{fig:formation}
    \end{subfigure}
    \hfill
    \begin{subfigure}{0.49\columnwidth}
        \centering
        \includegraphics[scale=0.9]{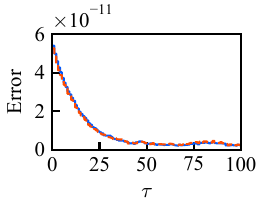}
        \caption{}
        \label{fig:error}
    \end{subfigure}
    \caption{(\subref{fig:formation}) Trajectories of the agents in the original formation control. (\subref{fig:error}) Control input errors $\| \tilde{u}(\tau) - \hat{u}(\tau) \|_{\max}$ (blue solid line) and $\| u(\tau) - \hat{u}(\tau) \|_{\max}$ (red dashed line).}
    \label{fig:example}
    \vspace{-1.5\baselineskip}
\end{figure}

\section{Conclusion}\label{sec:conclusion}

We proposed a one-round two-party computation protocol for the approximate matrix multiplication of fixed-point numbers based on the LWE and SIS problems.
We also established sufficient conditions for fixed-point number precision to guarantee the desired approximation error bound.
When applied to linear control laws, the protocol achieved lower online computational complexity for the client than the original computation, thereby mitigating the client burden of encrypted controls.
Future work will extend this framework to dynamic controllers and optimize the computational costs for the computing parties.

\section*{Acknowledgements}

The author used Gemini 3~\cite{Gemini} throughout the manuscript to improve its grammar and clarity.

\bibliographystyle{IEEEtran}
\bibliography{reference}

\end{document}